%% file: root.tex
\documentclass[letterpaper, 10 pt, conference]{ieeeconf}  

\IEEEoverridecommandlockouts                              

\usepackage{graphics} 
\usepackage{amsmath} 
\usepackage{amssymb}  
\usepackage{comment, epstopdf, accents, algorithm}
\usepackage[hidelinks]{hyperref}
\usepackage{multirow}
\usepackage{subcaption, graphicx}
\usepackage{lipsum}
\usepackage{algpseudocode, bm}
\makeatletter
\def\BState{\State\hskip-\ALG@thistlm}
\makeatother

\algnewcommand\True{\textbf{true}\space}
\algnewcommand\False{\textbf{false}\space}
\algnewcommand\Continue{\textbf{continue}\space}
\algnewcommand\Not{\textbf{not}\space}

\newcommand{\ubar}[1]{\underaccent{\bar}{#1}}
\newtheorem{theorem}{Theorem}
\newtheorem{lemma}{Lemma}

\newtheorem{problem}{Problem}

\newtheorem{definition}{Definition}

\title{\LARGE \bf
	Resilient Supervisory Control of Autonomous Intersections in the Presence of Sensor Attacks
}

\author{Amin Ghafouri and Xenofon D. Koutsoukos}

\begin{document}

\maketitle
\thispagestyle{empty}
\pagestyle{empty}

\begin{abstract}
	Cyber-physical systems (CPS), such as autonomous vehicles crossing an intersection, are vulnerable to cyber-attacks and their safety-critical nature makes them a target for malicious adversaries. This paper studies the problem of supervisory control of autonomous intersections in the presence of sensor attacks. Sensor attacks are performed when an adversary gains access to the transmission channel and corrupts the measurements before they are received by the decision-making unit. We show that the supervisory control system is vulnerable to sensor attacks that can cause collision or deadlock among vehicles. To improve the system resilience, we introduce a detector in the control architecture and focus on stealthy attacks that cannot be detected but are capable of compromising safety. We then present a resilient supervisory control system that is safe, non-deadlocking, and maximally permissive, despite the presence of disturbances, uncontrolled vehicles, and sensor attacks. Finally, we demonstrate how the resilient supervisor works by considering illustrative examples.
\end{abstract}

\input{intro}
\input{system}
\input{sensor}
\input{stealthy}
\input{result1}
\input{result2}
\input{result3}
\input{result4}
\input{experiment}
\input{discussion}
\bibliographystyle{abbrv}
\bibliography{reference}

\end{document}

%% file: intro.tex
\section{Introduction}\label{sec1}
Cyber-physical systems (CPS) are vulnerable to cyber-attacks and their safety-critical nature makes them a target for malicious adversaries. In this paper, we study the problem of supervisory control of autonomous intersections in the presence of sensor attacks. Sensor attacks are performed when an adversary gains access to the transmission channel and corrupts the measurements before they are received by the decision-making unit. Such attacks can put the system at risk by interfering with the operation of the control system. 

Recent advances in transportation technology indicate that autonomous intersections will be possible in the near future. The aim of autonomous intersections is to provide a safe, scalable, and efficient framework for coordinating autonomous vehicles. In the literature, there exist numerous protocols for autonomous intersections \cite{dresner:multiagent, azimi:reliable, kowshik:provable}. Among these approaches, studying the problem of autonomous intersection in the context of supervisory control of discrete event systems has received increasing attention since it allows incorporating the continuous dynamics and formally analyzing system safety \cite{colombo:supervisory, dallal:discrete, dallal:imperfect, ahn:supervisory}.

Supervisory control for autonomous intersections has three requirements: (1) safety, i.e., collisions must be avoided; (2) non-blockingness, i.e., vehicles should not deadlock; and (3) maximal-permissiveness, i.e., vehicles must not be restricted unless necessary \cite{cassandras:introduction}. A supervisory controller that satisfies these requirements is designed for a set of controlled and uncontrolled vehicles considering an unknown disturbance \cite{dallal:discrete}. This is done by creating a finite discrete event system abstraction of the continuous system and formulating the problem in the context of supervisory control for discrete event systems. Further, the supervisor is reconstructed to handle bounded measurement uncertainties \cite{dallal:imperfect}.

In this paper, we study the resilience of the system when deception attacks are performed on sensor measurements by a malicious omniscient adversary \cite{cardenas:survivable}. The problem is related to supervisory control with imperfect measurements \cite{dallal:imperfect}. However, sensor attack can be designed strategically and may have different effects from imperfect measurements. 




Concerns about resiliency of supervisory control systems is not recent as it has been addressed in the area of fault-tolerant control \cite{sampath:failure, zhao:monitoring}. An approach for monitoring and diagnosis of hybrid systems that integrate model-based and statistical methods is proposed in \cite{zhao:monitoring}. Further, a framework for fault-tolerant control of discrete event systems is designed that ensures the system recovers from any fault \cite{wen:framework}. However, the focus is on known faults rather than unknown and carefully designed cyber-attacks. 

Cyber-security of control systems has been studied for different classes of systems. It is shown that knowledge of the physical system can be used for attack detection by considering malicious attacks on process control systems \cite{cardenas:process}. Denial-of-service attacks are studied for discrete-time linear dynamical systems \cite{amin:safe}. Further, detection limitations for CPS are characterized in \cite{pasqualetti:attack}. In particular, it is shown that an attack is undetectable if the measurements due to the attack coincide with the measurements due to some nominal operating condition. 

Our objective is to study the problem of resilient supervisory control of autonomous intersections in the presence of sensor attacks. We introduce a detector in the control architecture with the purpose of detecting deception attacks on sensors. This detector incorporates knowledge of the physical system with the previously received data and detects a wide range of sensor attacks. Nevertheless, there exist stealthy attacks that cannot be detected but are capable of compromising safety. To address this issue, we design a resilient supervisor that maintains safety even in the presence of stealthy attacks. The resilient supervisor consists of an estimator that computes the smallest state estimate compatible with the control inputs and measurements seen thus far. We prove that the estimated set contains the true state of the system even if the measurements are corrupted and we show that the set of estimated states depends on the detector threshold. We then formulate the resilient supervisory control problem and present a solution. We demonstrate how the resilient supervisor works by considering illustrative examples. 

The remainder of this paper is organized as follows. Section \ref{sec2} defines the system model and sensor attacks, followed by formulation of the problem. In Section \ref{sec3}, we introduce a detector in the system architecture and characterize stealthy attacks. We design the resilient supervisor in Section \ref{sec4}. In Section \ref{sec5}, we illustrate how the resilient supervisor works by providing some examples. We conclude the paper in Section \ref{sec6} with discussion and future work.

%% file: system.tex
\section{System Model}\label{sec2}
\subsection{Plant}
We use the system model presented in \cite{dallal:discrete}. Consider a set $\mathcal{N} = \{1, ..., n\}$ of vehicles driving towards an intersection. The vehicles are modeled as single integrators and their dynamics are described by
\begin{equation}\label{ct}
	\dot{x} = v + d
\end{equation}
where $x\in X \subset \mathbb{R}^n$ is the state (i.e., position), $v \in V \subset \mathbb{R}^n$ is the control input (i.e., speed), and $d \in D \subset \mathbb{R}^n$ is a disturbance input representing unmodeled dynamics. Assume $v \in V$ is a vector with elements in the finite set $\{\mu a, \mu(a+1), ..., \mu b\}$, where $\mu \in \mathbb{R}^+$ and $a, b \in \mathbb{N} $. We refer to $a\mu$ and $b\mu$ as $v_{min}$ and $v_{max}$. The disturbance signal is bounded by $[d_{min}, d_{max}]^n$, and $v_{min} + d_{min} \geq \mu$, so that $\mu$ is the lowest possible speed. We assume that a subset of the vehicles are uncontrolled. To represent uncontrolled vehicles, we partition the vector $v$ into two subvectors $v_c \in V_c$ and $v_{uc} \in V_{uc}$, such that $v = (v_c, v_{uc})$ and $V = V_c \times V_{uc}$, where $v_c$ denotes the inputs of the controlled vehicles and $v_{uc}$ denotes the inputs of the uncontrollable vehicles.

Assume the control input is kept constant over each time interval $[k\tau, (k+1)\tau]$. Define the time $\tau$ discretization of the system \eqref{ct} as
\begin{equation}\label{discrete}
x_{k+1} = x_k + u_k + \delta_k
\end{equation}
where $ x_k = x(k\tau)$, $u_k=v(k\tau)\tau $, and $\delta_k = \int_{k\tau}^{(k+1)\tau}d(t)dt$. Define $U = V\tau$ and $\Delta = D\tau$, and we have $u\in U$ and $\delta \in \Delta$. Also, let $U = U_c \times U_{uc}$, where $U_c$ is the set of inputs for the controllable vehicles and $U_{uc}$ is the set of inputs for the uncontrollable vehicles. Denote by $u = (u_c, u_{uc})$ the actions of the controllable and uncontrollable vehicles for any $u \in U$.

\begin{figure}
	\begin{center}
		\includegraphics[width=5.5cm]{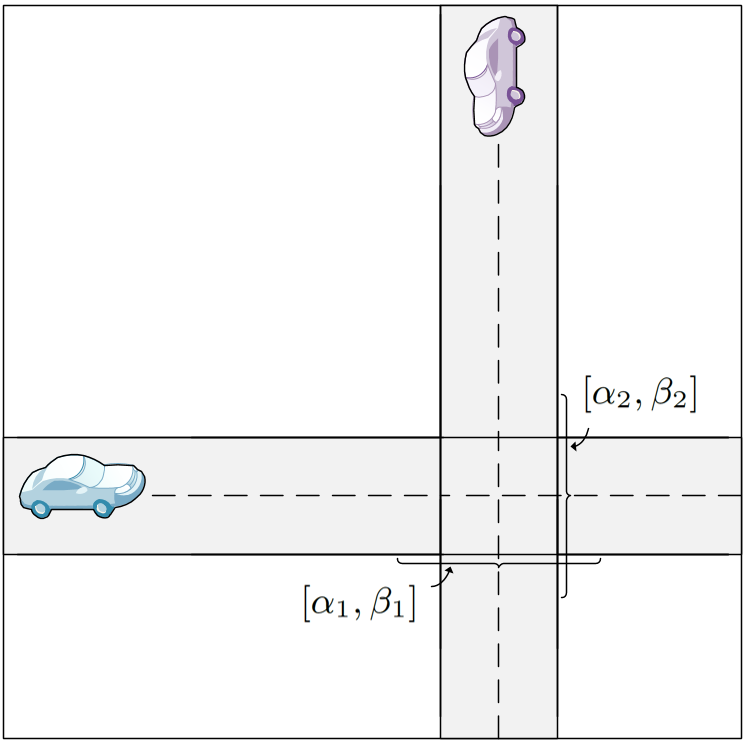}
		\caption{Two-vehicle autonomous intersection}
		\label{fig:int}
	\end{center}
\end{figure}


For each road $l$, let $[\alpha_l, \beta_l] \subset \mathbb{R}$ denote the part of the road that is in the intersection as shown in Fig. \ref{fig:int}. Let $i \in \Pi_l$, if vehicle $i$ drives along road $l$. A collision occurs between two vehicles $i \in \Pi_l$ and $j \in \Pi_{l'}$ with $ l \not = {l'}$, whenever $x_i \in [\alpha_l, \beta_l] \; \textrm{and} \; x_j \in [\alpha_{l'}, \beta_{l'}]$ simultaneously. Further, assuming $\gamma \in \mathbb{R}$ is the minimum safe distance between vehicles driving on the same road, a rear-end collision occurs when $|x_i - x_j| < \gamma$, $x_i \leq \beta_l$, $x_j \leq \beta_l$, and $i,j \in \Pi_l$. Calling the set of all the collision points the bad set $B \subset X$, no collision occurs if the system never reaches a state in $B$. A trajectory $x(t)$ of \eqref{ct} is defined safe, if:
\begin{equation}\label{safety}
\inf_{t \geq 0, b'\in B} \: || x(t)-b'||_{\infty} > 0
\end{equation}

%% file: sensor.tex
\subsection{Sensor Attacks and Problem Formulation}
Sensor attacks are performed when an adversary compromises some sensors and corrupts the measurements before they are received by the decision-making unit \cite{cardenas:survivable}. Such attacks, which compromise integrity of the measurements, may lead to collision or deadlock among vehicles. To model sensor attacks, we assume $T_a=[k_s,k_e)$ is the period of attack and $e_i(k\tau) \in E \subseteq \mathbb{R}$, is an error added to sensor $i$ at timestep $k \in T_a$. Defining $e(k\tau)$ as the vector $(e_1(k\tau),...,e_n(k\tau))$, the corrupted measurement vector $\tilde{x}$ is described by
\begin{equation}\label{error}
\tilde{x}(k\tau)=x(k\tau)+ e(k\tau), \; k \in T_a
\end{equation}

We assume there exists $T_{max} \in \mathbb{R}^+$ such that  $|T_a| \leq T_{max}$. Because of the non-blockingness property, vehicles eventually cross the intersection, and therefore scenarios involving specific cars and intersections have a finite-time horizon. Since the sensor attack is launched on individual vehicles, the maximum attack duration, $T_{max}$, can be assumed to be finite. This variable can be selected as the maximum time it takes for a single vehicle to cross the intersection, which can easily be computed given the minimum speed (i.e., $v_{min} + d_{min}$) and the set of initial states. It should be noted that this implies an assumption that the initial state of a car is not corrupted. Although this assumption is restrictive, it is reasonable since detection of cars approaching intersections can be performed also with other means than the vehicles sensors. Also, if this assumption is not true an adversary can add "zombie" cars (i.e., cars that do not exist) or make existing cars "invisible", in which case, the detection and control problems considered in the paper are not well-defined.


Our goal is to design a supervisory controller that maintains the system's operational goals even in the presence of sensor attacks. These operational goals are: (1) safety, i.e., collisions must be avoided; (2) non-blockingness, i.e., vehicles must eventually cross the intersection; and (3) maximal-permissiveness, i.e., vehicles must not be restricted unless necessary \cite{dallal:discrete}. 

%% file: stealthy.tex
\section{Attack Detection}\label{sec3}

To be resilient to sensor attacks, a common approach is to employ detection mechanisms that detect attacks before they can cause significant damage. 



\subsection{Detector}
We use a nonparametric Cumulative sum (CUSUM) statistic as our detection method \cite{basseville:detection}. The detector's goal is to determine whether a sequence of vectors $z(1), z(2), ..., z(N)$ is generated through the system being under normal behavior ($H_0$) or attack ($H_1$). Similar to \cite{cardenas:process}, to formulate the detector and to avoid making limiting assumptions on the attacker's behavior, we assume the expected value of a random process $Z(k)$ that generates $z(k)$ is less than zero in the case of normal behavior, and greater than zero in the case of attack.

We denote by $\textbf{Post}_{u_c} x$, the set of all the reachable states given a state $x$ and an input $u_c$, that is, $\textbf{Post}_{u_c} (x) := \cup_{u_{uc}\in U_{uc}, d \in D} \; (x + u + d)$. Given a measurement $\tilde{x}_i(k)$, the parameter $z_i(k)$ is defined for every vehicle $i$ as
\begin{equation}\label{detectorz}
z_i(k):=\inf_{\hat{x}_i(k) \in \textbf{Post}_{u_i}(\tilde{x}_i(k-1))} \: ||\tilde{x}_i(k) - \hat{x}_i(k)|| - b_i
\end{equation}
where $b_i \in \mathbb{R}$ is a small constant chosen in a way that $E_{H_0}[z_i(k)]<0$. 

The CUSUM statistic is described by:
\begin{equation}\label{detectorstat} 
C_i(k)= (C_i(k-1) + z_i(k))^+
\end{equation}
where $C_i(0)=0$, and $(a)^+ = a $ if $a \geq 0 $.  

Assigning $\eta_i$ as the threshold chosen based on a desired false alarm rate, the corresponding decision rule is defined as
\begin{equation}\label{detectordecision}
d(C_i(k)) = \left\{ \begin{array}{rcl}
H_1 & \textrm{  if } C_i(k)>\eta_i \\ H_0 & \textrm{otherwise} \\
\end{array}\right.
\end{equation}

We assume that upon detection (i.e., $d(C_i(k)) = H_1$ for a $k$), the system switches to a mode where vehicles are either driven manually or controlled by a fail-safe controller. 

\subsection{Stealthy Attacks}
Although the detector can identify attacks that cause a large difference between received measurements and predicted states (i.e., detectable attacks), a skilled attacker may attempt to evade this scheme via performing \textit{stealthy attacks} \cite{cardenas:process}. This is done by keeping the detector's output below the threshold of detection and not triggering an alert. Stealthy attacks exist in our system because of the following factors:\\
\textit{(i) Detector's threshold}: As shown in Fig. \ref{fig:stealthyspace}, small thresholds (i.e., $\eta \rightarrow 0)$ reduce the size of stealthy attacks space. This means a system designer should always select the smallest possible threshold to achieve maximum resiliency. However, employment of detectors with small thresholds results in high false alarm rates in realistic noisy environments. This is why thresholds are often chosen non-zero (i.e., $\eta \not = 0$), which allows an adversary to launch stealthy attacks.\\
\textit{(ii) Disturbance and uncontrolled vehicles}: Disturbance and uncontrolled vehicles add non-deterministic behavior to the system. This makes it impossible to perform exact state estimation. Hence, as long as the measurements due to the attack cannot be differentiated from the measurements due to some normal operation, the attack remains stealthy \cite{pasqualetti:attack}.

\begin{figure}
	\begin{center}
		\includegraphics[width=5.8cm]{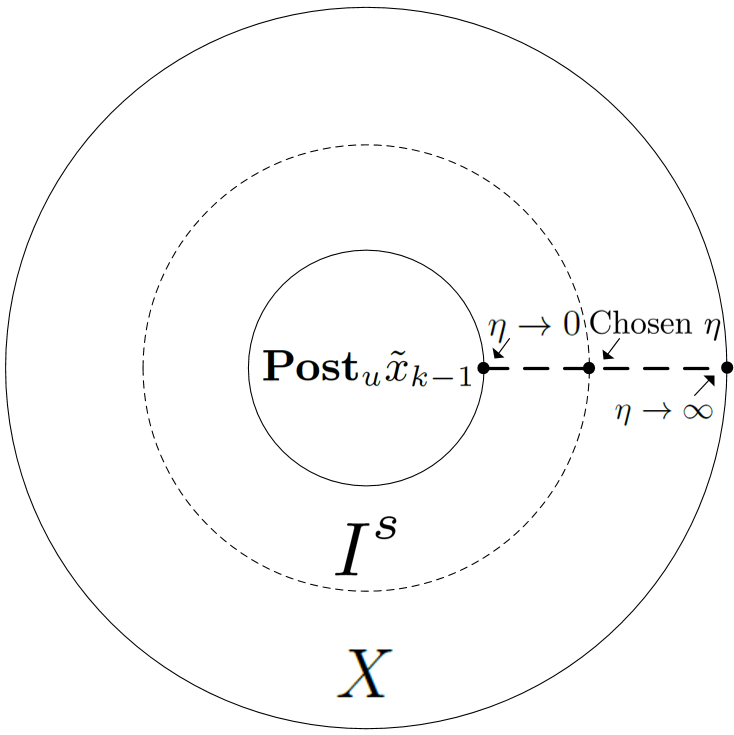}
		\caption{Choosing a small threshold for the detector reduces size of stealthy attack space $I^s$.}
		\label{fig:stealthyspace}
	\end{center}
\end{figure}



Stealthy attacks must be characterized and countermeasures for securing the system against them must be designed. To do so, we characterize stealthy attacks for the detector \eqref{detectordecision}. We do this by deriving a condition for error vector $e_k$ given the actual state $x(k\tau)$, which is only known by the attacker. We define $\hat{x}_{max}(k\tau) := \max \textbf{Post}_{u_{k-1}}(\tilde{x}_{k-1})$ and $\hat{x}_{min}(k\tau) := \min \textbf{Post}_{u_{k-1}}(\tilde{x}_{k-1})$ \footnote{We temporarily abuse notation and exclude $\tau$ when referring to a time $k\tau$.}.

\begin{lemma}\label{stealthyattack}
	(Stealthy Attacks): A sensor attack with $e_T \in E^{n\times|T|}$ is stealthy iff for all $k\in T_a$ we have $e_k \in [\ubar{e}_k, \bar{e}_k]$ where
	\begin{equation*}
	\ubar{e}_k(x(k), C(k-1)) = - x(k) + \hat{x}_{min}(k) - \eta - b + C(k-1),
	\end{equation*}
	\begin{equation*}
	\bar{e}_k(x(k), C(k-1)) = - x(k) + \hat{x}_{max}(k) + \eta + b - C(k-1).
	\end{equation*}
	
\end{lemma}

\begin{proof}
	We show a corrupted vector $\tilde{x}(k)$ bypasses the detector if and only if it is within the stated bound. Let $\tilde{x}(k) \in X$ be a stealthy attack $\Leftrightarrow$ $C(k) < \eta$ $\Leftrightarrow$ $C(k-1) + z(k) < \eta$ $\Leftrightarrow$ $C(k-1) + \inf ||\tilde{x}(k) - \hat{x}(k)|| - b < \eta$ $\Leftrightarrow$ for each $i$ we have $\inf |\tilde{x}_i(k) - \hat{x}_i(k)| < \eta_i + b_i - C_i(k-1)$ $\Leftrightarrow$ $\inf |x_i(k) + e_i(k) - \hat{x}_i(k)| < \eta_i + b_i - C_i(k-1)$. Solving this, we have $ - x_i(k) + \hat{x}_{i,min}(k) -\eta_i - b_i + C_i(k-1) < e_i(k) < - x_i(k) + \hat{x}_{i,max}(k) + \eta_i + b_i - C_i(k-1)$ $\Leftrightarrow$ $\ubar{e}  < e_{k} < \bar{e} $.
\end{proof}

\begin{figure}
	\begin{center}
		\includegraphics[width=8.4cm]{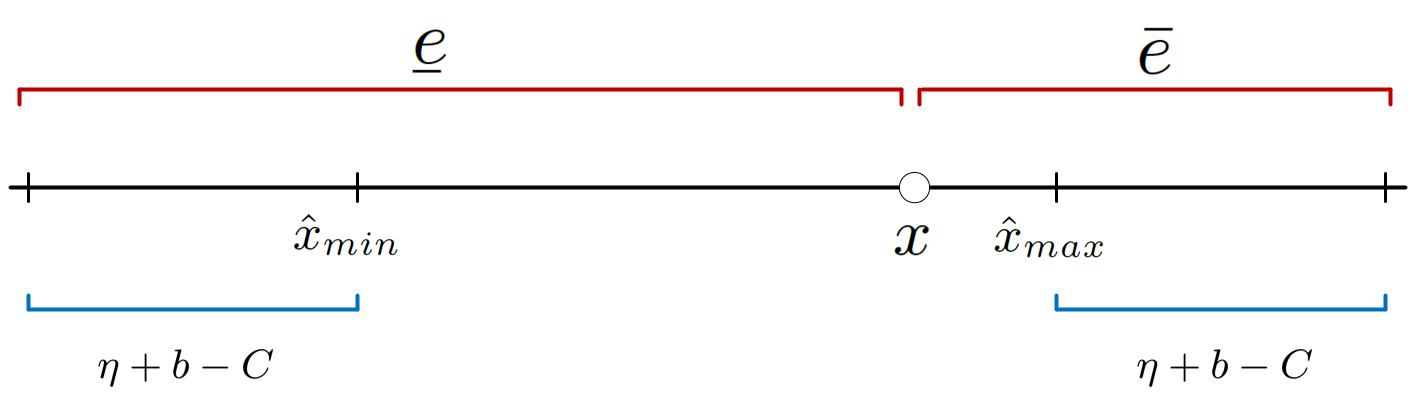}
		\caption{Sensor attack is stealthy if $\tilde{x} \in [\hat{x}_{min} - (\eta + b - C), \hat{x}_{max} + \eta + b -C]$.}
		\label{fig:stealthy}
	\end{center}
\end{figure}


The above lemma characterizes stealthy attacks by deriving a condition for error vector $e_k$ assuming the actual state is known. We use \eqref{error} and Lemma \ref{stealthyattack} to define a set of stealthy attacks $I^s_k \subseteq X$, which contains all the corrupted measurements $\tilde{x}_k$ that are stealthy.




\begin{lemma}\label{lemmastealthyset}
	(Stealthy Attacks Space): A sensor attack with $e_T \in E^{n\times|T|}$ is undetectable iff for all $k \in T_a$, the measurement $\tilde{x}_k$ belongs to stealthy attacks set $I^s_k: X \times U_c \times \mathbb{R} \rightarrow 2^X$ described by
\begin{multline}\label{stealthy}
	I^s_k(\tilde{x}_{k-1},u_{k-1}, C_{k-1})= \\ [\hat{x}_{min,k} - \eta - b + C, \hat{x}_{max,k} + \eta + b -C]
\end{multline}
\end{lemma}

%% file: result1.tex
\section{Resilient Supervisory Control System}\label{sec4}
In this section, we design a resilient supervisory control system that is secure against stealthy attacks. This is done by taking the following steps:
(1) Designing an estimator system that computes the smallest state estimate compatible with the sequence of control inputs and possibly corrupted measurements received;
(2) Formulating the supervisory control problem given the estimator system with respect to a state space discretization;
(3) Solving the problem by creating a discrete event system (DES) abstraction of the estimator system, translating the control problem to the abstracted domain, solving it, and translating the results back to the continuous domain.

\subsection{Estimator}
We present a method to compute the smallest state estimate compatible with the sequence of control inputs and possibly corrupted measurements received. We first use the assumption on maximum duration of sensor attacks to find a set that surely contains the actual state of the system even if measurements are corrupted. We then present a prediction-correction scheme based on the work in \cite{dallal:imperfect}, which computes a set of predicted states and then corrects it. Finally, we design an estimator system that receives as inputs the measurements, the control inputs, and the detector's statistics and computes the smallest state set that contains the actual state. 

\begin{figure}
	\begin{center}
		\includegraphics[width=8.4cm]{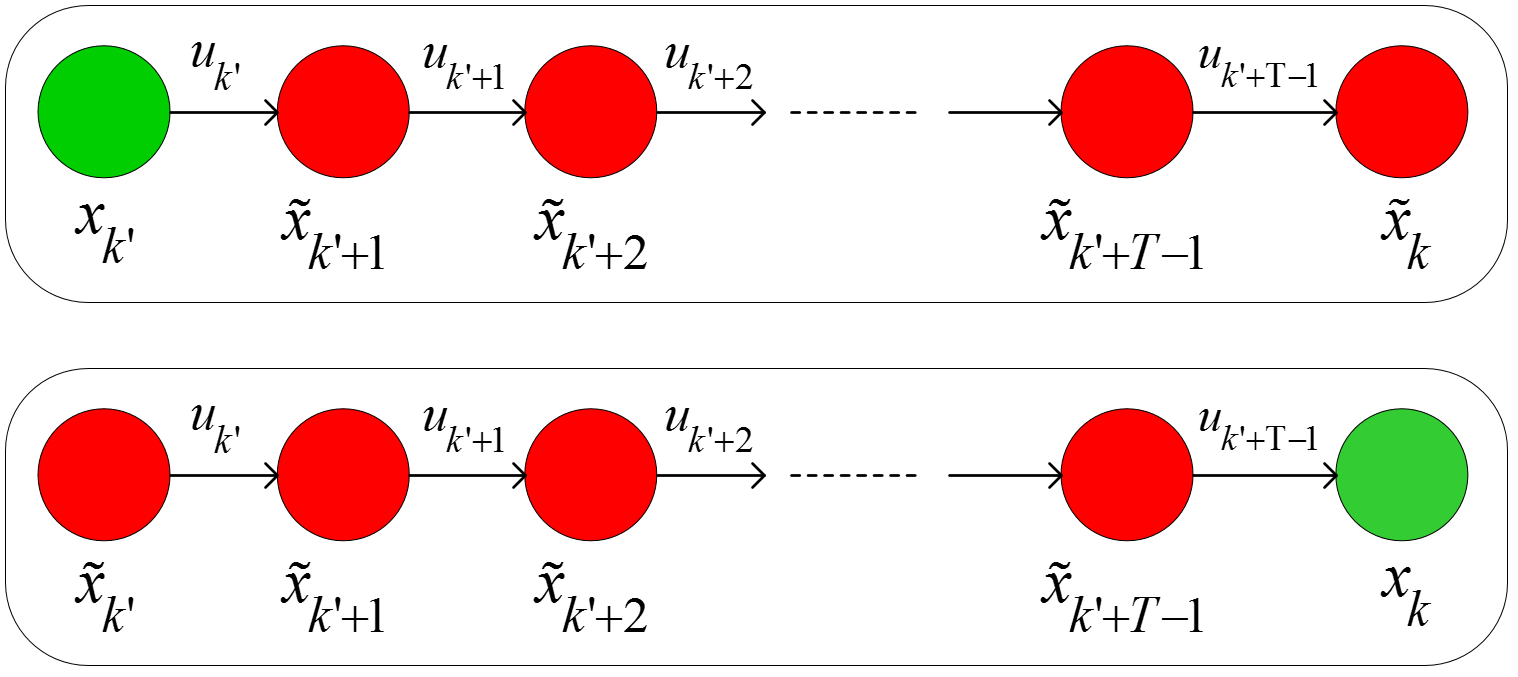}
		\caption{Assuming $k = k' + T_{max}$, either $\tilde{x}_{k'}$ or $\tilde{x}_{k}$ is not attacked.}
		\label{fig:bounded}
	\end{center}
\end{figure}

In the following, we assume $k' = k - T_{max}$ and we define an extension of post-operator as $\textbf{Post}_{u_{c,1}...u_{c,m}}x = x + \sum_{i = 1}^{i = m}u_{c,i} + \bigcup_{u_{uc,i} \in U_{uc}, \delta_i\in\Delta} \sum_{i = 1}^{i = m} u_{uc,i} + \delta_i$. 

\begin{lemma}\label{trustfinder}
	Given measurements $\tilde{x}_{k}$ and $\tilde{x}_{k'}$, and previously chosen controllable inputs $(u_{c,k'}, ..., u_{c,k-1})$, output of function $\hat{I}_k: X \times \underbrace{ U_c \;\; ... \;\; U_c}_{\displaystyle T_{max}} \times X \rightarrow 2^X$ defined by
	\begin{equation}\label{Ihat}
	\hat{I}_k(\tilde{x}_{k'}, u_{k'},...,u_{k},\tilde{x}_{k})=\{ \emph{\textbf{Post}}_{u_{c,k'}...u_{c,k-1} }\tilde{x}_{k'} \} \cup \tilde{x}_{k}
	\end{equation}
	contains the actual state of the system for any timestep $k \in \mathbb{N}$, that is, $x_k \in \hat{I}_k, \; \forall k \in \mathbb{N}$.
\end{lemma}

\begin{proof}
	 We directly use the assumption that length of attacks is at most $T_{max}$. As shown in Fig. \ref{fig:bounded}, this implies that $\tilde{x}_{k'}$ and $\tilde{x}_{k}$ may not be corrupted at the same time. Let us consider both cases. If $\tilde{x}_{k}$ is corrupted, $T_{max}$ steps before (i.e., $k - T_{max} = k'$), the received measurement is uncorrupted, that is to say, $\tilde{x}_{k'} = x_{k'}$. Hence, we can write $x_k \in \textbf{Post}_{u_{k'}...u_{k-1}} x_{k'}\subseteq \hat{I}_k$. Next, suppose $\tilde{x}_{k'}$ is corrupted, which means the measurement received at timestep $k = k' + T_{max}$ is accurate, i.e., $\tilde{x}_{k} = x_k$. This is followed by $x_k \in \hat{I}_k$. Thus, $\hat{I}_k$ contains the true state of the system in both cases. 
\end{proof}


Lemma \ref{trustfinder} holds for any two measurements that have a time distance of at least $T_{max}$. This gives rise to definition of the prediction-correction scheme below that recursively computes set of predicted states $I^p \subseteq 2^X$ and set of corrected states $I^c \subseteq 2^X$. We assume the initial state set $X_0$ is not in the bad region. Also, for $k < T_{max}$, we define $\hat{I}_k = X$. 
\begin{equation*}\label{Ic0}
	I_{0}= \{ x \in X_0  \}
\end{equation*}
\begin{equation}\label{Ip1}
	I^{p}_k(I^c_{k-1}, u_{c,k-1})= \bigcup_{x \in I^c_{k-1} , u_{uc}\in U_{uc}, \delta \in\Delta} x + u + \delta
\end{equation}
\begin{equation}\label{estimator}
	I^c_{k} = I^{p}_k \cap \hat{I}_k
\end{equation}

Using the sets defined above and the set of stealthy attacks \eqref{stealthy}, we design an estimator system that receives as inputs the measurements, the control inputs, and the detector's statistics and computes the smallest state set that contains the actual state. 

\begin{definition}
	(\textit{Estimator System}): 
	Estimator system corresponding to system \eqref{ct} employed with the detector \eqref{detectordecision} is
	\begin{equation}\label{I} 
		I_{k}(I^p_k,I^s_k,\tilde{x}_k) = \left\{ \begin{array}{llc} I^c_k & \tilde{x}_k \in I^s_k \\ \textrm{undefined} &  \textrm{else} \end{array}\right.
	\end{equation}
\end{definition}


Note that the set of stealthy attacks $I^s_k$ describes how the detector \eqref{detectordecision} affects the estimator system. In particular, if a measurement $\tilde{x}$ is not in the set of stealthy attacks, i.e., $\tilde{x}_k \not \in I^s_k$, the output of the estimator system is undefined because a detection alert is triggered. On the other hand, if $\tilde{x}$ belongs to the set of stealthy attacks, i.e., $\tilde{x}_k \in I^s_k$, the output of the estimator is the corrected set $I^c_k$ defined by \eqref{estimator}. 

\textit{Remark}: There is a tradeoff between detector's threshold and estimator performance. More specifically, if the detector's threshold increases, the size of stealthy attack space increases as well (Lemma \ref{lemmastealthyset}), which leads to a larger set of estimated states (Equation \eqref{trustfinder}). Consequently, more restriction needs to be enforced by the supervisory controller to handle this larger set of states.



\subsection{Supervisory Control Problem}

To formulate the supervisory control problem, we first transform the infinite state space of the system model to a finite set of discrete states. Similar to \cite{dallal:discrete}, we define a set of discrete states $Q$ and a mapping $\ell : X \rightarrow Q$ as follows:
\begin{equation}\label{ell}
\ell_i(x_i) := \left\{ \begin{array}{ccc}
c\tau\mu, \textrm{ for } c\in \mathbb{Z} \textrm{ s.t. } \\ c\tau\mu - \tau\mu/2 < x_i \leq c\tau\mu + \tau \mu / 2 & x_i \leq \beta_l \\ q_{i,m} &  x_i > \beta_l
\end{array}\right.
\end{equation}
where $l$ is the index of vehicle $i$'s road $(i \in \Pi_l)$. We define $\ell(x)$ as the vector $(\ell(x_1), ...,\ell(x_n))$. Further, we assume that for all $q \in Q$, there exists some $x\in X$ such that $\ell(x) = q$.

Let $2^X/\ell$ denote the quotient set of $2^X$ by the equivalence classes induced by $\ell$. The supervisory control problem is stated below:


\begin{problem}\label{problemdes}
	Given $2^X/\ell$, design a supervisor $\sigma: 2^{X}/\ell \rightarrow 2^{V_c}$ that associates to each $I_k \in 2^X$ a set of inputs $v_c \in V_c$ allowed for the interval $[k\tau,(k+1)\tau]$ with the following properties:
\textit{(i}) if $v_c(t) \in \sigma(I_{k})$, then the trajectory of the system is safe in the interval $[k\tau,(k+1)\tau]$.
\textit{(ii}) if $v_c \in \sigma(I_{k})$ and $\ell(I_{k+1}) \neq {q_m}$, then $\sigma(I_{k+1}) \neq \emptyset$.
\textit{(iii}) if $\tilde{\sigma} \neq \sigma_R$ and $\tilde{\sigma}$ is non-blocking and safe, then $\tilde{\sigma}(I_{k}) \subseteq \sigma_R(I_{k})$.
\end{problem}

%% file: result2.tex
\subsection{Resilient Supervisor Design}
We construct a DES abstraction of the estimator system and solve the supervisory control problem using the abstracted model. More specifically, we take the following steps:
(1) Defining a DES with a set of events corresponding to controlled inputs, uncontrolled inputs, disturbance, prediction-correction, and detection such that the uncontrolled inputs and the disturbance events are unobservable.
(2) Constructing an observer of the DES and providing conditions under which the observer is an abstraction of the estimator system
(3) Translating safety and non-blockingness specifications to the DES domain, solving the supervisory control problem given the observer, and translating the obtained supervisor to the continuous time domain.
\subsubsection{DES Model}

We define a DES with an event set that models controlled inputs, uncontrolled inputs, disturbance, prediction-correction, and detection. This is similar to the approaches used in \cite{dallal:discrete} and \cite{dallal:imperfect}. Consider a five-layer event set $E = \Lambda^d \times \Lambda^c \times U_c \times U_{uc} \times W$. The event sets $U_c$ and $U_{uc}$ contain inputs of controlled and uncontrolled vehicles and are defined as described in Section \ref{sec2}. The set $W$ represents a discretization of the set of disturbances $\Delta$, defined by $W = \{k\tau\mu: \lfloor\delta_{min}/(\tau\mu)\rfloor \leq k \leq \lceil\delta_{max}/(\tau\mu)\rceil \}$. The sets $\Lambda^c$ and $\Lambda^d$, which respectively represent prediction-correction and detection, are as defined below.

$\boldsymbol{\Lambda^c}$: Let any two inputs of $\hat{I}(\cdot)$ be equivalent, if they result in the same discrete set with respect to $\ell(\hat{I}(\cdot))$. That is to say, given \eqref{Ihat} and \eqref{ell}, define the equivalence relation $(x_1, u_1...u_{T-1}, x_T) \equiv (x_1', u_1'...u_{T-1}', x_T')   \Leftrightarrow \\ \ell(\hat{I}(x_1, u_1...u_{T-1}, x_T)) = \ell(\hat{I}(x_1', u_1'...u_{T-1}', x_T'))$. Define $\phi := (x_1, u_1...u_{T-1}, x_T)$ and $\Phi := X \times U^{T} \times X$. For any $\phi \in \Phi$, let $[\phi]$ denote the equivalence class of $\phi$ under relation $\equiv$ and let
$\Lambda^c$ represent the set of equivalence classes. These equivalence classes form a finite set of discrete events $\lambda^c \in \Lambda^c$.

We say upon occurrence of event $\lambda^c$, there exists a transition from a state $q$ to itself if and only if there exist some continuous element $\phi$ whose equivalence class is $\lambda^c$ \cite{dallal:imperfect}. To represent this, define the discrete function $I^{\Lambda^c}(\lambda^c) : \Lambda^c \rightarrow 2^{Q}$ as
\[ I^{\Lambda^c}(\lambda^c) = \{ q \in Q : (\exists \phi \in \Phi : [\phi] = \lambda^c)[\hat{I}(\phi) \cap \ell^{-1}(q) \neq \emptyset] \} \]

Then, define the transition function $\psi^c : Q \times \Lambda^c \rightarrow Q$, which is shown to be the analogue of the estimate correction \eqref{estimator} in the DES domain, as
\begin{equation*}
\psi^c(q, \lambda^c) = \left\{ \begin{array}{ll}
q & \textrm{  if } q \in I^{\Lambda^c}(\lambda^c)  \\ \textrm{undefined} & \textrm{else} \\
\end{array}\right.
\end{equation*}


$\boldsymbol{\Lambda^d}$: Let any two elements $(x, u, C)$ and $(x', u', C')$ be in the same equivalence class if their sets of discrete states with respect to $\ell(I^s(\cdot))$ are the same. That is, given \eqref{stealthy}, let \\$ (x, u, C) \equiv (x', u', C') \Leftrightarrow \ell(I^s(x, u, C)) = \ell(I^s(x', u', C'))$. Define $\theta := (x, u, C)$, and $\Theta := X \times U \times \mathbb{R}$. For any $\theta \in \Theta$, let $[\theta]$ denote the equivalence class of $\theta$ under this relation. Denote by $\Lambda^s$ the set of such equivalence classes. The finite set $\Lambda^d$ is defined by $\Lambda^d := \Lambda^s \times Q$.



To define the transition function corresponding to $\Lambda^d$, define the function $I^{\Lambda^d} : \Lambda^s \times Q \rightarrow \{H_0,H_1\}$, for any $\tilde{q} \in Q$, by
\begin{equation*}
I^{\Lambda^d}(\lambda^s, \tilde{q}) = \left\{ \begin{array}{ccl}
H_0 & \exists \theta \in \Theta:[\theta] =\lambda^s, I^s(\theta) \cap \ell^{-1}(\tilde{q}) \neq \emptyset  \\ H_1 & \textrm{else} \\
\end{array}\right.
\end{equation*}

 The transition function $\psi^d : Q \times \Lambda^d \rightarrow Q$, which corresponds to the detector in the DES domain, is defined as
\begin{equation*}
\psi^d(q, \lambda^d) = \left\{ \begin{array}{lll}
q & \textrm{ if } I^{\Lambda^d}(\lambda^d) = H_0  \\ \textrm{undefined} & \textrm{ else } \\
\end{array}\right.
\end{equation*}
where $I^{\Lambda^d}(\lambda^d) \not = H_0$ represents detection of attacks. 

Define the DES
\begin{equation}\label{myDES}
\tilde{G} := (\tilde{Q}, E, \psi , q_0, Q_m)
\end{equation}
with state set $\tilde{Q}$, event set $E$, transition function $\psi: Q \times E \rightarrow Q$, set of initial states $Q_0 \subset Q$, and set of marked states $Q_m = \{q_m\}$. Transition function $\psi$ has the form $\psi(q,\lambda^d,\lambda^c,u_c,u_{uc},w) = \\ \psi_3(\psi_2(\psi_1(\psi^c(\psi^d(q,\lambda^d),\lambda^c),u_c),u_{uc}),w)$, and the language of this system is $(\Lambda^d\Lambda^c U_cU_{uc}W)^*$. To avoid state conflicts, we define intermediate states $Q'$, $Q''$, $Q_{I1}$, and $Q_{I2}$ such that $\psi^d: Q \times \Lambda^d \rightarrow Q'$, $\psi^c: Q' \times \Lambda^c \rightarrow Q''$, $\psi_1: Q'' \times U_c \rightarrow Q_{I1}$, $\psi_2: Q_{I1} \times U_{uc} \rightarrow Q_{I2}$, and $\psi_3: Q_{I2} \times W \rightarrow Q$. The uncontrolled inputs $U_{uc}$ and the disturbances $W$ are not directly observed so we define them as uncontrollable and unobservable. Further, we let the controllable inputs $U_c$ be observable and controllable, and $\Lambda^c$ and $\Lambda^d$ be observable but uncontrollable \cite{dallal:imperfect}.

%% file: result3.tex
\subsubsection{Observer}

System $\tilde{G}$ is a discrete event system with unobservable events and we can design its observer $Obs(\tilde{G})$ \cite{cassandras:introduction}. The states of the observer are information states $\iota \subseteq 2^{\tilde{Q}}$ and its transition functions are: $\bar{\psi}^d: 2^Q \times \Lambda^d \rightarrow 2^{Q'}$, $\bar{\psi}^c: 2^{Q'} \times \Lambda^c  \rightarrow 2^{Q''}$, and $\bar{\psi}^p: 2^{Q''} \times U_c \rightarrow 2^Q$. Note that by construction, all controllable events of $Obs(\tilde{G})$ are observable and thus, an optimal supervisor solution indeed exists.

We show that if $\delta_{min}$ and $\delta_{max}$ are multiples of $\mu$, the DES observer automaton constructed above is an abstraction of \eqref{I} and its controller can be used to control the continuous system. We use the notions of simulation and alternating simulation \cite{tabuada:verification}. In what follows, consider the observation maps $H_{2^X}(I) = \ell(I)$ and $H_{2^Q}(\iota) = \iota$, and the relation $R = \{(\iota,I) \in 2^Q \times 2^X:\iota = \ell(I) \}$.

\begin{lemma}\label{GsimI}
	If $\delta_{min}$ and $\delta_{max}$ are multiples of $\mu$, system \eqref{I} alternatingly simulates $Obs(\tilde{G})$, and $Obs(\tilde{G})$ simulates system \eqref{I}.
\end{lemma}
\begin{proof}
	(System \eqref{I} alternatingly simulates $Obs(\tilde{G})$): 1) By assumption, for every $\iota_0 \subseteq 2^{Q_0}$ there exists $I_0 \subseteq 2^X_0$ with $\ell(I_0) = \iota_0$;	2) By definition of the observation maps, for every $\ell(I) = \iota$ we have $H_{2^Q}(\iota) = H_{2^X}(I)$; 3) We need to prove for every $\iota = \ell(I)$ and for every $u \in U_{2^Q}(\iota)$ there exists $u' \in U_{2^X}(I)$	such that for every $(I,u,I')$	there exists $(\iota,u',\iota')$	satisfying $\ell(I') = \iota'$. Consider any transition $(I, u, I')$ and let $u = u'$. The sets $I$ and $\iota$ can be written as $I = [x_{l}, x_{h}]$ and $\iota = \ell(I) = \{[q_{l}, q_{h}]\}$. Because $\delta_{min}$ and $\delta_{max}$ are multiples of $\mu$, we have $\bar{\psi}^p(\iota, u) = \ell(I^p)$ \cite{dallal:cps}. Then, $\iota' = \bar{\psi}^c(\bar{\psi}^d(\bar{\psi}^p)) = \ell(I^p) \cap \ell(\hat{I}) = \ell(I^p \cap \hat{I}) = \ell(I')$.
	
	($Obs(\tilde{G})$ simulates system \eqref{I}): 
	1) By assumption, for every $I_0 \subseteq 2^{X}$ there exists $\iota_0 \subseteq 2^Q$ with $\ell(I_0) = \iota_0$;
	2) It is similar to alternating simulation;
	3) Let $u = u'$. We know if $\ell(I) = \iota$, then $\ell(I^p) = \bar{\psi}^p(\iota,u_c)$. By definition, $\ell(I') = \ell(I^p \cap \hat{I}(\phi)) $ and by construction there exist $\lambda^c=[\phi]$ and $\lambda^s=[\theta]$ such that $ \ell(I^p \cap \hat{I}) = \bar{\psi}^c(\bar{\psi}^d(\bar{\psi}^p,\lambda^d),\lambda^c) = \iota'$. Hence, for every $(I,u,I')$ there exists $(\iota,u,\iota')$ such that $\ell(I')=\iota'$. \end{proof}

%% file: result4.tex
\subsubsection{Supervisor Solution}
To construct a safe, non-blocking, and maximally permissive supervisor at the DES level, we use the solution to the Basic Supervisory Control Problem in the Non-Blocking case (BSCP-NB) as described in \cite{ramadge:supervisory,cassandras:introduction}. The solution
computes the supremal controllable sublanguage of a specification with respect to the language of a system. More specifically, it computes the language $(\mathcal{L}_m(H))^{\uparrow C}$, where $H$ is a specification and $\uparrow C$ is the supremal controllable sublanguage operation \cite{koutsoukos:supervisory}. The BSCP-NB algorithm constructs a supervisor $S$ such that
$\mathcal{L}_m(S/G) = (\mathcal{L}_m(H))^{\uparrow C}$ and $\mathcal{L}(S/G) = \overline{(\mathcal{L}_m(H))^{\uparrow C}}$,
where $S/G$ is the system $G$ controlled by $S$ and $\overline{L}$ denotes the prefix closure of language $L$ \cite{cassandras:introduction}.

We define the transition function, safety specification, and set of marked states for the estimator system and then translate it to DES domain \cite{dallal:imperfect}. The estimator system \eqref{I} defines a transition system as follows:
\begin{multline*}
	(I,u,I')\in \rightarrow \subseteq 2^X \times U_c \times 2^X \textrm{ if } \\ \exists \tilde{x} \in X, I(I^p,I^s,\tilde{x})=I'
\end{multline*}

Define a transition from $(I,u,I')$ as safe if all trajectories of the continuous time system corresponding to this transition are safe:
\begin{multline}\label{safetyspec}
	(I,u,I')\in Safe \subseteq \rightarrow \textrm{ if } \\ \forall x \in I, x' \in I', (x,u,x') \in \rightarrow_{\eqref{discrete}} \Rightarrow (x, u, x') \in Safe_{\eqref{discrete}}
\end{multline}
Further, define the set of marked states as:
\begin{equation}\label{marked}
	I_m := \{I \subseteq 2^X: \ell(I) = {q_m} \}
\end{equation}

The safety specification \eqref{safetyspec} and the set of marked states \eqref{marked} are translated to the DES domain as follows:
\begin{multline}\label{SafeG}
	(\iota, u, \iota') \in Safe_{\tilde{G}} \subseteq 2^Q \times U \times 2^Q \textrm{ if } \\ \forall (I,u,I') \in \rightarrow, \iota = \ell(I), \iota' = \ell(I'): (I,u,I') \in Safe
\end{multline}
\begin{equation}\label{QmG}
	\iota_m := \{\{q_m\}\}
\end{equation}



Solving BSCP-NB for $Obs(\tilde{G})$ with respect to the specifications \eqref{SafeG} and \eqref{QmG}, outputs a safe, non-blocking, and maximally permissive supervisor \cite{cassandras:introduction}. Denote by $S$ this safe, non-blocking, and maximally permissive supervisor. Define the supervisor map
\begin{equation}\label{supervisormap}
	\sigma(I) = \{ u_c / \tau : u_c \in S(\ell(I)) \}
\end{equation}

We call \eqref{supervisormap}, the resilient supervisory controller of system \eqref{ct}.


\begin{theorem}
	The supervisory controller \eqref{supervisormap} that associates a set of admissible inputs $v_{k} \in V_c$ to each $x_k \in X$, allowed for the interval $t = [k\tau, (k+1)\tau]$ solves Probem \ref{problemdes}.
\end{theorem}




%% file: experiment.tex
\section{Example}\label{sec5}

We consider two controllable vehicles driving on separate roads approaching an intersection. Consider the set of inputs $V_c = \{1, 3\}$, disturbances $d = [0,1]^2$, time-discretization $\tau = 1$, space discretization parameter $\mu = 1$, maximum attack length $T_{max} = 1$, and detector threshold $\eta \rightarrow 0$. Let $x_0 = (1,1)$, and suppose $(\alpha_1, \beta_1) = (\alpha_2, \beta_2) = [9.5 , 12.5]$, so the vehicles collide if at a time $t$, $9.5 \leq x_1(t), x_2(t) \leq 12.5$. A marked state is reached if $x > 12.5$ for both vehicles. 

We define an attack model known as surge attack. In such attacks, the adversary adds the maximum error value that bypasses the detector, that is to say, he keeps the detector's statistic at the threshold without exceeding it \cite{cardenas:process}. This required solving the equation $C(k) + z(k) - b = \eta$, for every $k \in [k_s, k_e)$. Using Lemma \ref{lemmastealthyset}, the resulting attack at $k=k_s$ is
\[ \tilde{x}_{k_s} = \hat{x}_{max, k_s} + \eta + b\]
where $\hat{x}_{max, k_s} = x_{k_s-1} + u_{k_s-1} + d_{max}$. For the remainder of the attack, which is $k \in [k_s + 1, k_e)$, we have
\[ \tilde{x}_{k} = \tilde {x}_{k-1} + u_{k-1} + d_{max} + b \]
Note that we can replace $\hat{x}_{max}$ and $d_{max}$ with $\hat{x}_{min}$ and $d_{min}$ to achieve similar results for negative error values.

To show how sensor attacks compromise safety, we first consider the supervisor designed in \cite{dallal:discrete} as the controller of the intersection. Constructing the supervisor, the admissible inputs for each cell $q \in Q$ are shown in Fig. \ref{fig:diag}. Note that since $n = 2$, the bad set $B$ is bounded. The set of enabled inputs for the initial state is $\tilde{\sigma}(x(0))=\{(1,3),(3,1)\}$. Suppose the vehicles select $u = (1,3)$ and the state evolves to $x(1) = x(0) + u(0) + (0,0) = (2,4)$. The attacker receives the actual measurement $x(1)$ and starts a surge attack on vehicle 2 by transmitting the corrupted vector $\tilde{x}(1) = x(1) + (0, 1) = (2, 5)$. The detector will not trigger an alert because $\tilde{x}(1) \in I^s_1$. The set of admissible inputs for the corrupted state is $\tilde{\sigma}(\tilde{x}(1)) = \{(1,1),(1,3)\}$, whereas for the actual state this set is $\tilde{\sigma}(x(1)) = \{(1,3)\}$. The input $(1,1)$ is incorrectly enabled by the supervisor in the perceived state and hence, selecting it may result in deadlock (i.e., unavoidable collision). 

\begin{figure}
	\begin{center}
		\includegraphics[width=8.6cm]{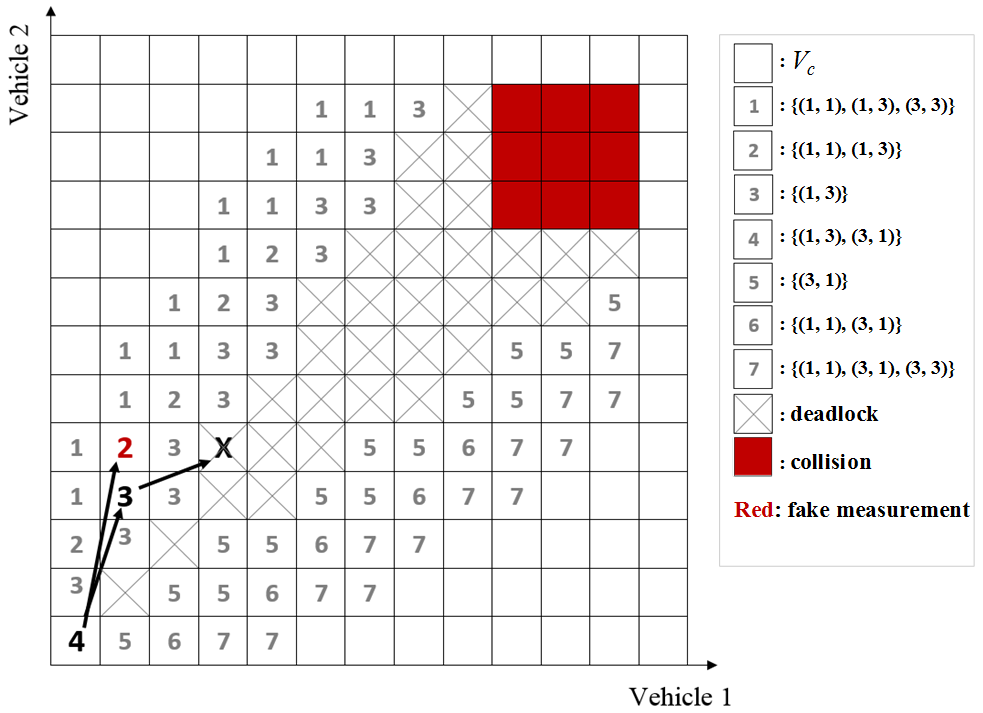}
		\caption{Example of a stealthy sensor attack that compromises safety.} 
		\label{fig:diag}
	\end{center}
\end{figure}

Next, we construct the resilient supervisor for the example and show how the discussed stealthy attack is handled by the supervisor. To implement the resilient supervisor, we need to construct the observer of system $\tilde{G}$ as shown in Fig. \ref{fig:examplesup}. The figure only shows the different choices of inputs for vehicle 2 and input of vehicle 1 is set to $v_1 = 1$. Similar to the previous case, let $u(0) = (1,3)$ be the first chosen input so the information state of the system becomes $(\{2,3\}\{4,5\})$. Assume the actual state is $x(1) = (2,4)$ and the attacker transmits the corrupted measurement $\tilde{x}(1) = (2,5)$. The events $\lambda^c = [\tilde{x}(0),u(0),\tilde{x}(1)]$ and $\lambda^d = [\tilde{x}(0),u(0),C(0)]$ occur and the information state transits to $\iota(1) = (\{2,3\}\{4,5\})$. The input $(1,1)$, which caused collision in the previous case, is disabled in this scenario since if selected, the vehicles may enter any of the illegal states $(4,5), (5,5)$, and $(5,6)$. Therefore, another input is selected (e.g., $u(1) = (1,3)$) and the vehicles remain safe despite the attack.


\begin{figure}
	\begin{center}
		\includegraphics[width=8.8cm]{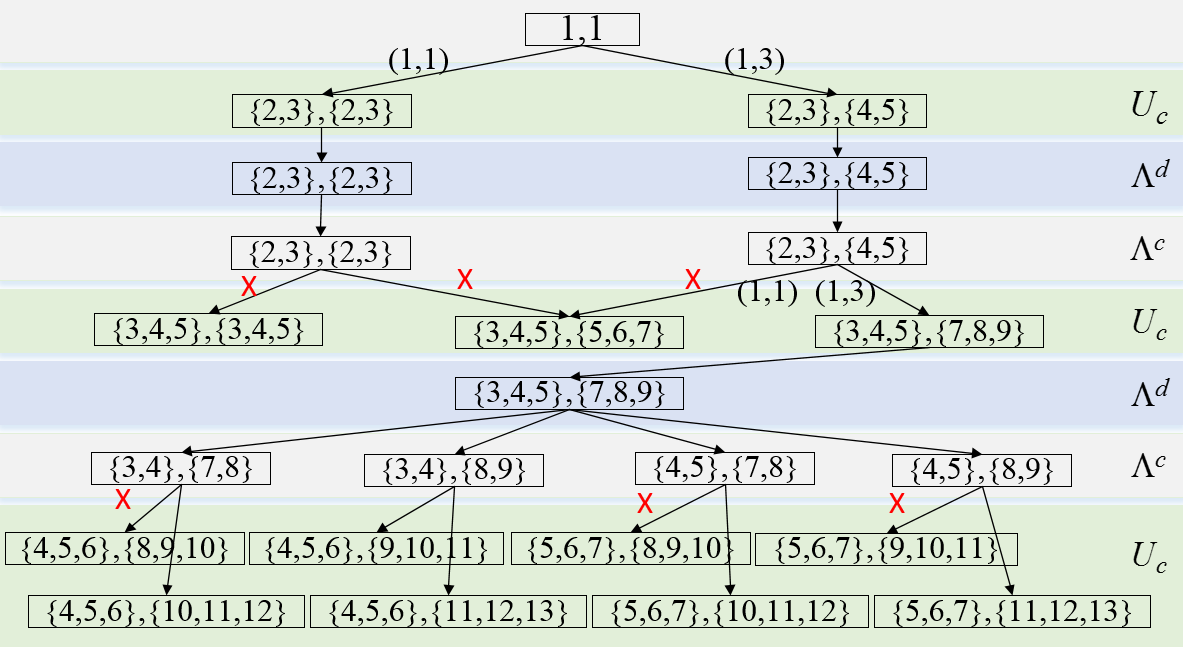}
		\caption{Observer of the two-vehicle intersection.} 
		\label{fig:examplesup}
	\end{center}
\end{figure}

To show how the state evolves in the following steps, we construct the observer for an additional sequence of events as shown in Fig. \ref{fig:examplesup}. Assume the current information state is $(\{3,4,5\}, \{7,8,9\})$. Upon occurrence of an event $\lambda^c$, this information state has 4 possible successors. For instance, if the new measurement is $\tilde{x}(2) = (3,8)$, then $\lambda^c = [(\tilde{x}(1),u(1),\tilde{x}(2))]$ with $\tilde{x}(1) = (2,5)$, $u(1) = (1,3)$, and $\tilde{x}(2) = (3,8)$, and the successor state is $\iota(2) = (\{3,4\},\{8,9\})$. This shows how the measurements $\tilde{x}(1)$ and $\tilde{x}(2)$, and the input $u(1)$ are used to correct the predicted set of states to a subset $\iota(2)$.

To simulate the scenarios, we use Simulation of Urban MObility (SUMO). SUMO is an open source platform for microscopic road traffic simulation \cite{krajzewicz:recent}. The objects of the SUMO simulation can be controlled through Traffic Control Interface (TraCI), which uses a TCP-IP based client/server architecture. We use the TraCI4Matlab implementation of TraCI. 


The simulations and their corresponding trajectories are shown in Fig. \ref{fig:sumo}. We assume vehicle 1 drives the same way in both scenarios. Also, we assume following an unavoidable collision, vehicles keep driving with the same speed. The simulations confirm our previous discussion that the resilient supervisor maintains safety even in the presence of sensor attacks. In this case, this is done by disabling the unsafe input $(1,1)$ at $t = 1$, and enabling $(1,3)$ at subsequent steps.

\begin{figure}
	\centering
	\begin{subfigure}{0.185
			\textwidth}
		\includegraphics[width=\linewidth]{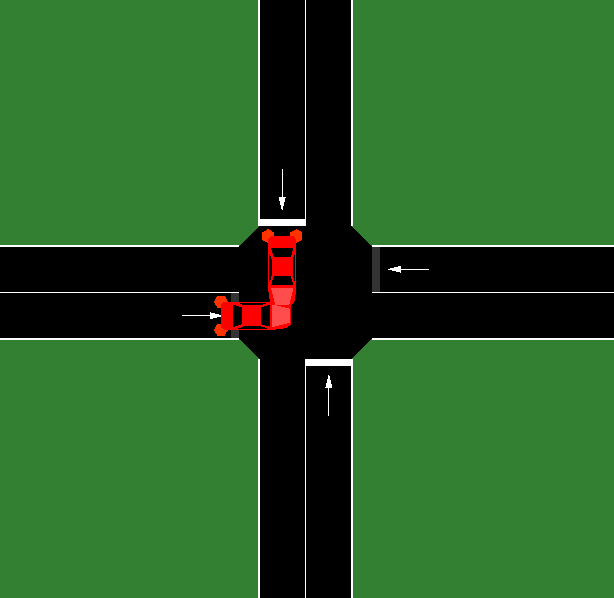}
		\caption{} \label{fig:ex1}
	\end{subfigure}
	~ 
	\begin{subfigure}{0.18\textwidth}
		\includegraphics[width=\linewidth]{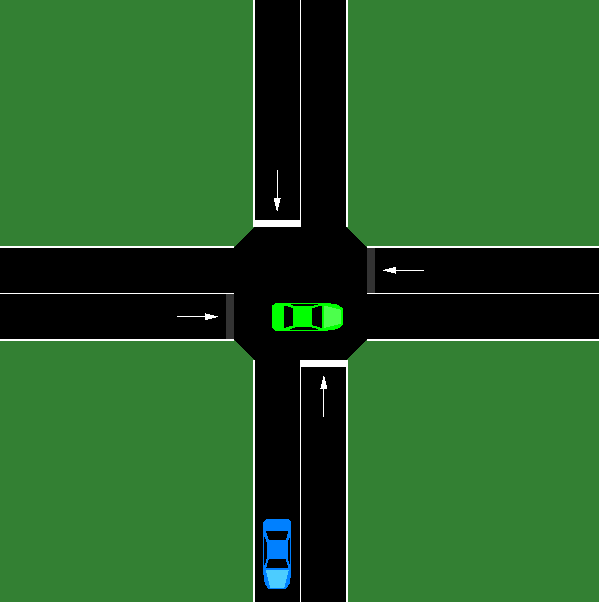}
		\caption{} \label{fig:ex2}
	\end{subfigure}
	~
	\begin{subfigure}{0.21\textwidth}
		\includegraphics[width=\linewidth]{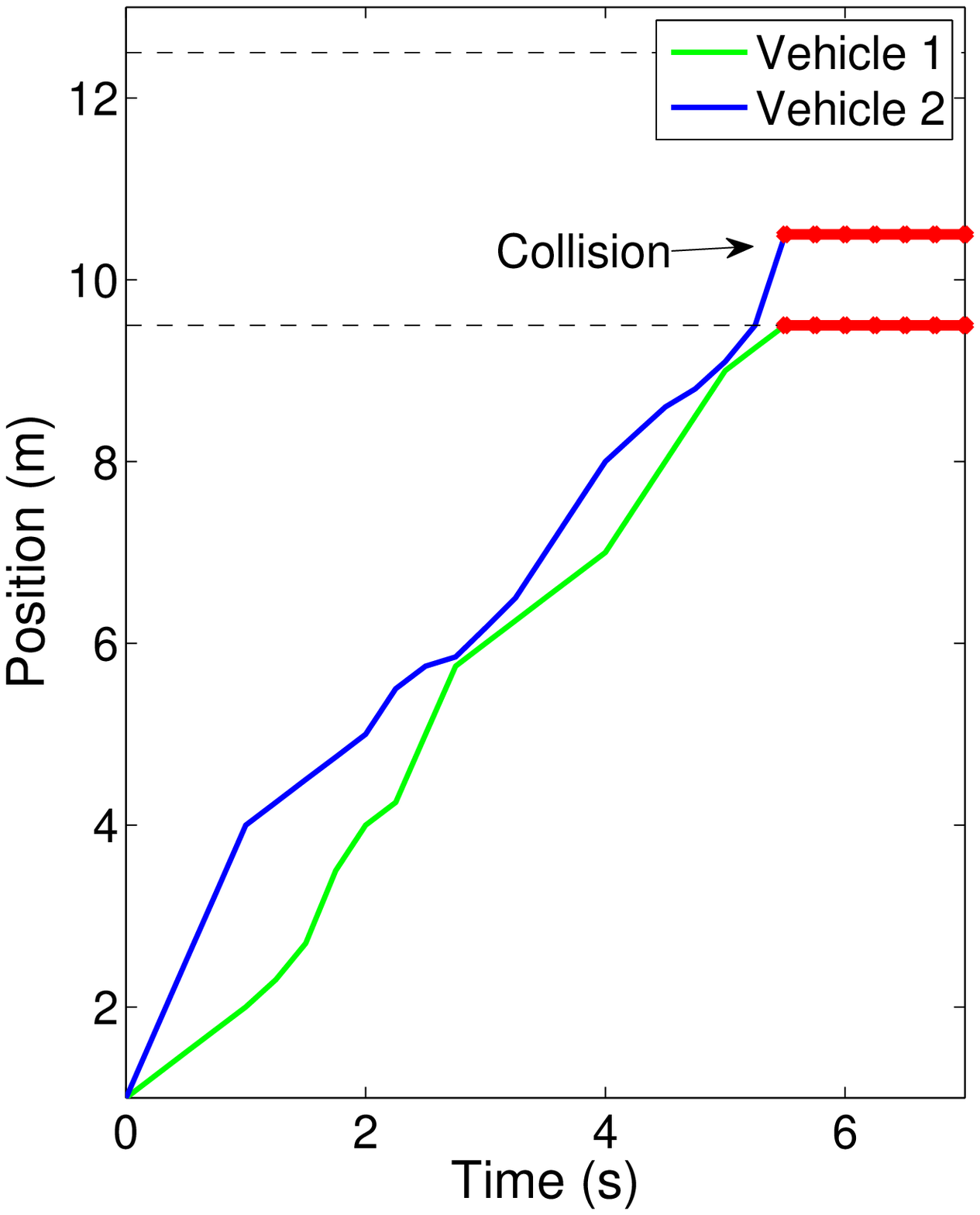}
		\caption{} \label{fig:ex1t}
	\end{subfigure}
	~
	\begin{subfigure}{0.21\textwidth}
		\includegraphics[width=\linewidth]{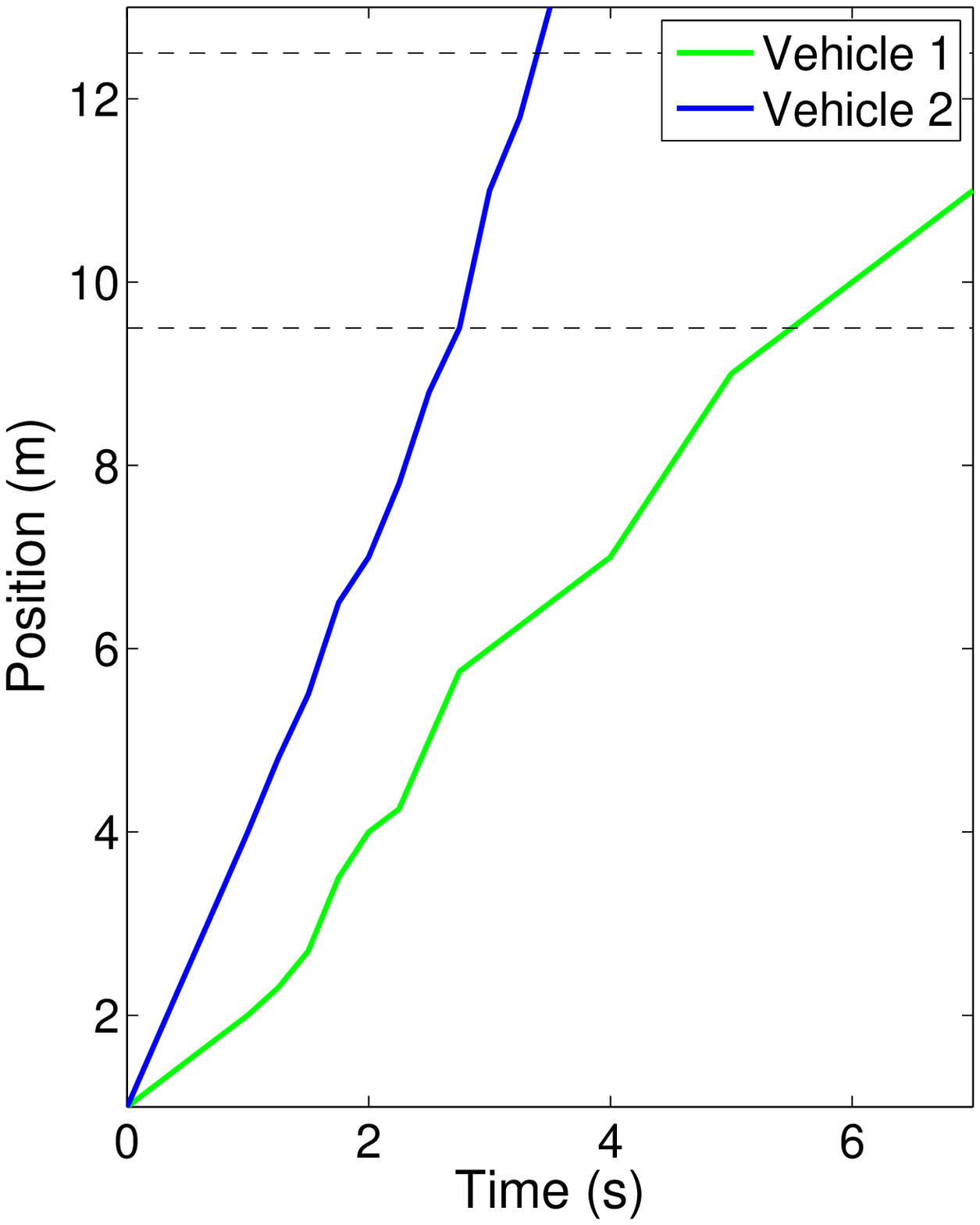}
		\caption{} \label{fig:ex2t}
	\end{subfigure}
	\caption{Simulation of the example in SUMO. Figures \ref{fig:ex1} and \ref{fig:ex1t} correspond to the non-resilient supervisor and Figures \ref{fig:ex2} and \ref{fig:ex2t} correspond to the resilient supervisor. The area between the dashed lines represents the intersection.} \label{fig:sumo}
\end{figure}




Although safety can be maintained for any finite $T_{max}$, the higher $T_{max}$, the more restrictive the supervisor becomes. In other words, there exists a tradeoff between $T_{max}$ and supervisor's performance. To demonstrate this, one can extend the example by considering different values of $T_{max}$. If $T_{max} \geq 4$, inputs $\{(1,3),(3,1)\}$ are enabled at the initial state. Then, until a marked state is reached, the supervisor disables all the inputs except the one chosen at the initial state. This is clearly more restrictive than the case of $T_{max}=1$. Nevertheless, for this example, even if $T_{max} = (x_m-x_0)/(v_{min} + d_{min}) = 12$, a solution exists.

Scenarios including multiple vehicles can be similarly analyzed by verifying safety for each pair of vehicles (as shown in Fig. \ref{fig:multiple} for a different network with $n=5$). The scalability of the proposed resilient supervisory control system depends on the number of vehicles, number of control inputs, minimum and maximum disturbances, discretization parameters, and maximum attack length. Scalability can be characterized based on the computational complexity for implementing supervisors for partially-observed DES \cite{yin:general}.

\begin{figure}
	\centering
	\begin{subfigure}{0.2
			\textwidth}
		\includegraphics[width=\linewidth]{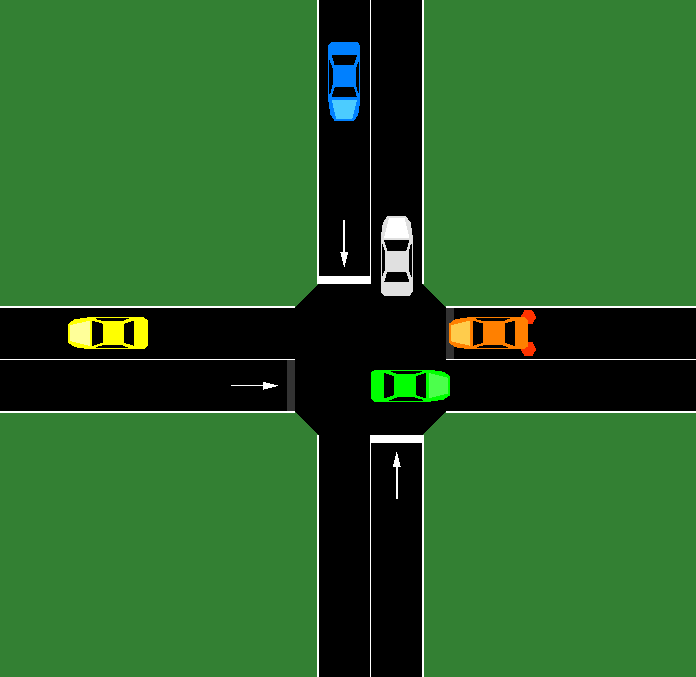}
	\end{subfigure}
	~
	\begin{subfigure}{0.2\textwidth}
		\includegraphics[width=\linewidth]{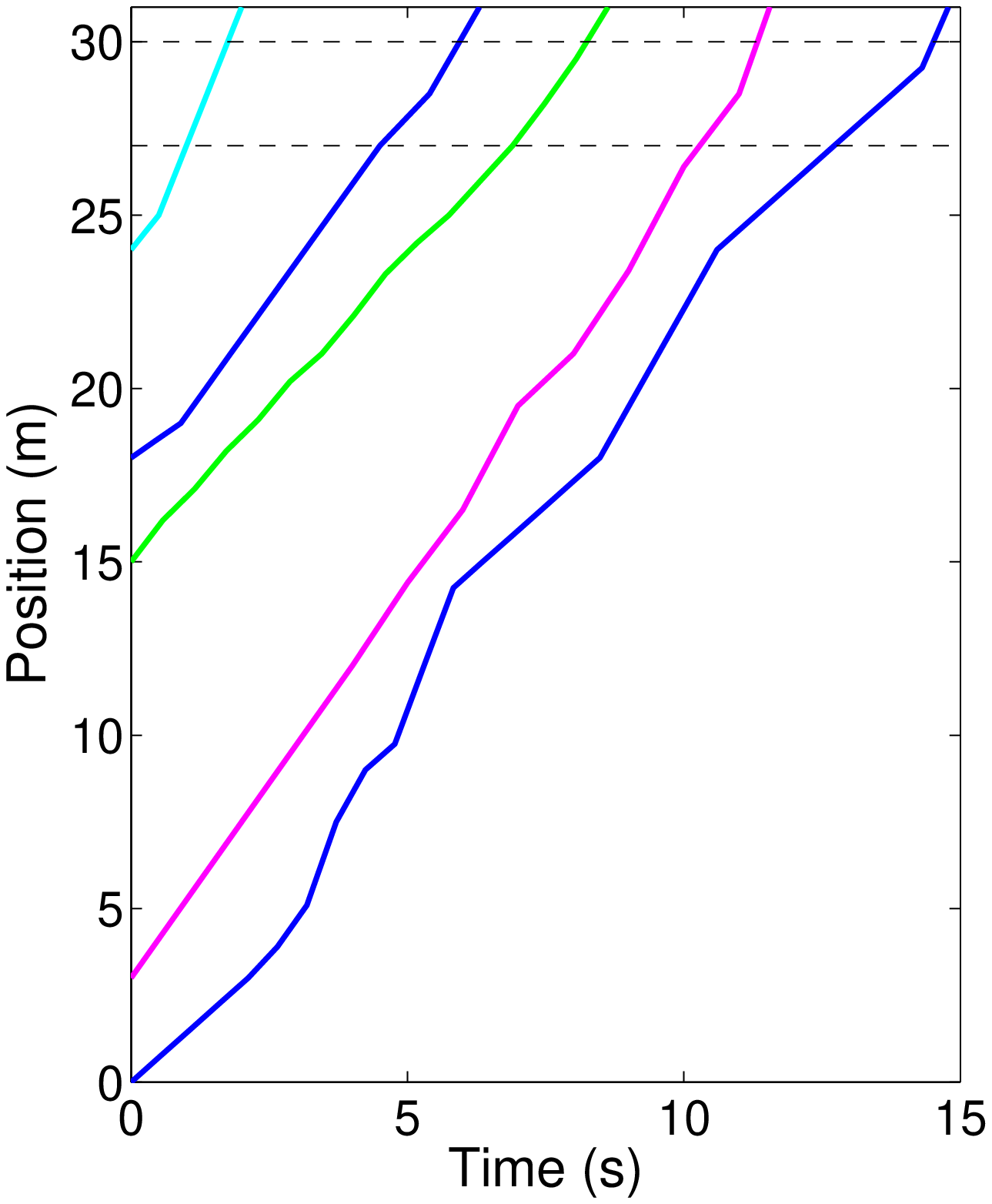}
	\end{subfigure}
	\caption{System remains safe if every pair of vehicles is safe.} \label{fig:multiple}
\end{figure}





%% file: discussion.tex
\section{Conclusion}\label{sec6}

We studied the problem of supervisory control of autonomous intersections in the presence of sensor attacks. We showed that the supervisory control system is vulnerable to sensor attacks that can cause collision or deadlock among vehicles. To improve the system resilience, we introduced a detector in the control architecture with the purpose of detecting attacks. There exist stealthy attacks that cannot be detected but are capable of compromising safety. To address this issue, we designed a resilient supervisor that maintains safety even in the presence of stealthy attacks. The resilient supervisor consists of an estimator that computes the smallest state estimate compatible with the control inputs and measurements seen thus far. We formulated the resilient supervisory control system and presented a solution. We demonstrated how the resilient supervisor works by considering examples.




Our solution can be extended to handle other attack models. Infinite length attacks and DoS attacks are special cases of the attack model discussed in this work and can be studied in a similar manner. Even though we focused on the analysis of supervisory control of autonomous intersections, our techniques can be applied to other cyber-physical system as well. Future work include study of actuator attacks, development of efficient methods for implementing resilient controllers, and construction of decentralized resilient controllers.